\documentclass{amsart}[11pt]
\usepackage{amsmath}
\usepackage{amssymb}
\usepackage{amsthm}
\usepackage{mathrsfs}
\usepackage[margin=1in]{geometry}
\usepackage{color}
\usepackage{graphicx}
\usepackage{amsfonts}
\usepackage{hyperref}
\usepackage{enumerate}
\usepackage{xcolor}
\usepackage{tikz}
\usetikzlibrary{decorations.pathreplacing, positioning}

\newtheorem{theorem}{Theorem}[section]
\newtheorem{corollary}[theorem]{Corollary}
\newtheorem{lemma}[theorem]{Lemma}
\newtheorem{proposition}[theorem]{Proposition}

\theoremstyle{definition}
\newtheorem{definition}[theorem]{Definition}
\newtheorem{remark}[theorem]{Remark}

\newcommand{\ind}{1\hspace{-2.1mm}{1}}

\newcommand{\pf}{\mathfrak{p}}
\newcommand{\B}{\mathfrak{B}}
\newcommand{\Binv}{\mathfrak{B}^{\leftarrow}}
\newcommand{\RR}{\mathbb{R}}
\newcommand{\PP}{\mathbb{P}}
\newcommand{\EE}{\mathbb{E}}
\newcommand{\K}{\mathrm{K}}
\newcommand{\D}{\mathrm{d}}

\newcommand{\Pp}{\mathcal{P}}
\newcommand{\Qq}{\mathfrak{V}}
\newcommand{\Qql}{\Qq^{\mathrm{loc}}}
\newcommand{\Vv}{\mathcal{V}}
\newcommand{\cs}{\mathfrak{c}}

\newcommand{\la}{\langle}
\newcommand{\ra}{\rangle}

\newcommand{\dm}{d_{-}}
\newcommand{\dpp}{d_{+}}
\newcommand{\vb}{\overline{v}}
\newcommand{\ts}{\tau}
\newcommand{\HH}{\mathbb{H}}
\newcommand{\BMO}{\mathrm{BMO}}
\newcommand{\Tt}{\mathcal{T}}
\newcommand{\Ct}{\widetilde{C}}
\newcommand{\Wt}{\widetilde{W}}
\newcommand{\E}{\mathrm{e}}
\newcommand{\QQ}{\mathbb{Q}}

\newcommand{\bs}{\BS}
\newcommand{\BS}{\mathrm{BS}}
\newcommand{\SVI}{\mathrm{SVI}}

\newcommand{\BSC}{\mathrm{C}^{\mathrm{BS}}}
\newcommand{\BSP}{\mathrm{P}^{\mathrm{BS}}}
\newcommand{\Nn}{\mathcal{N}}
\newcommand{\Ff}{\mathcal{F}}
\newcommand{\rrho}{\overline{\rho}}
\newcommand{\eps}{\varepsilon}
\newcommand{\bh}{\boldsymbol{\beta}}

\newcommand{\tht}{\theta_{t}}

\begin{document}

\title{Dynamics of symmetric SSVI smiles and implied volatility bubbles}
\author{Mehdi El Amrani}
\address{New York University}
\email{mehdielamrani95@gmail.com}
\author{Antoine Jacquier}
\address{Department of Mathematics, Imperial College London, and Alan Turing Institute}
\email{a.jacquier@imperial.ac.uk}
\author{Claude Martini}
\address{Zeliade Systems}
\email{cmartini@zeliade.com}
\thanks{We thank Stefano De Marco and Paolo Baldi for useful discussions.}

\date{\today}
\maketitle

\begin{abstract}
We develop a dynamic version of the SSVI parameterisation for the total implied variance,
ensuring that European vanilla option prices are martingales, hence preventing the occurrence of arbitrage,
both static and dynamic.
Insisting on the constraint that the total implied variance needs to be null at the maturity of the option, 
we show that no model--in our setting--allows for such behaviour.
This naturally gives rise to the concept of implied volatility bubbles, 
whereby trading in an arbitrage-free way is only possible during part of the life of the contract, 
but not all the way until expiry.
\end{abstract}

\section{Introduction}
Implied volatility is at the very core of financial markets, 
and provides a unifying and homogeneous quoting mechanism for option prices.
The literature abounds in stochastic models for stock prices that generate implied volatility 
smiles--with various degrees of practical success.
Among those, the Heston model~\cite{Heston} in equity and the SABR model~\cite{SABR} 
in interest rates--together with their ad hoc and in-house improvements--have been of particular importance.
Despite this success, these stochastic models do not enjoy the simplicity of closed-form expressions, 
and dedicated numerical techniques are needed to implement them.
One way to bypass this has been to consider approximations of option prices--and the corresponding implied volatilities--in asymptotic regimes;
thorough reviews of the latter are available in~\cite{FouqueBook, FGGJT}.
A different approach, pioneered by Gatheral~\cite{GatheralSVI}, consists in specifying 
a direct parameterisation of the implied volatility, 
having the clear advantage of speeding up computation and calibration times.
The original Stochastic Volatility Inspired (SVI) formulation, devised while its inventor was at Merrill Lynch, has proved extremely efficient in fitting volatility smiles on equity markets.
That said, it was only devised as a maturity slice interpolator and extrapolator, 
and different sets of parameters were needed in order to fit a whole surface (in strike and maturity).
Gatheral and Jacquier~\cite{SSVIGat} extended it to a whole surface,
devising tractable sufficient conditions ensuring absence of arbitrage. 
The design of calibration algorithms is then easy, and this SSVI formulation has been adopted widely in the financial industry, 
and has since been extended~\cite{corbetta2019robust, eSSVI} to a version with maturity-dependent correlation.

SSVI directly tackles option prices (equivalently, implied volatilities), without following the usual route of specifying a model for the evolution of the underlying. 
It is furthermore fully static, as its inputs are market option prices 
at a given point in calendar time, with only strike and expiry allowed to vary. 
Gatheral and Jacquier~\cite{SVIHeston} showed that, as the maturity increases, 
the SVI parameterisation was in fact the true limit of the Heston smile.
A natural question is thus whether there exists a dynamic model such that at each calendar time, 
the option smiles in this model are given by SVI or SSVI.

We provide here an answer, albeit through a slightly different lens,
as we investigate whether one can impose stochastic dynamics on the implied volatility,
ensuring that arbitrage cannot occur over time.
We work in a simplified and minimal setup, 
in a perfect market with no interest rates, in continuous time, 
and consider European Call/Put options with a fixed expiry, so we will restrict ourselves to the dynamic
of a fixed smile; we will also assume that the underlying process does not distribute coupons nor dividends and does not default. 
Motivated by the discussion above, we assume that the total implied variance
has an SSVI shape at all (calendar) times before maturity, but is allowed to move stochastically,
with the condition that both the underlying and option prices should be martingales. 

This is not the first suggested solution to this problem,
and several authors have attempted to propose joint dynamics for the underlying stock price
and the implied volatility.
Motivated by empirical evidence that the implied volatility moves over time, 
the usual approach is to specify a stochastic It\^o diffusion for the total implied variance, 
as in~\cite{carr2014implied, Hafner, Ledoit, Lyons, sun2014tale, SBucher}.
However, \textit{`the problem with market models is the extremely awkward
set of conditions required for absence of arbitrage'}~\cite{Davis}.
An important step was made by Schweizer and Wissel~\cite{schweizer2008arbitrage, schweizer2008term}, 
who derived general conditions ensuring existence of such market models.
Even if the resulting conditions are not easily tractable for modelling purposes, this result is the first positive answer. 
The only other positive result (for continuous processes) we are aware of can be found in Babbar's PhD thesis~\cite{Babbar},
which, building on~\cite{Lyons}, developed stochastic models for the joint stock price 
and the total implied variance (which she calls the operational time), for a fixed strike, 
relying on comparison theorems for Bessel processes. 
One fundamental catch, though, is that the implied volatility may hit zero strictly before the maturity of the option, making the model degenerate.
We shall revisit this degeneracy somehow, giving it some financial meaning.

The Implied Remaining Variance framework~\cite{carr2014implied, CarrWu}
shares a common point of view with our approach. 
There, the shape of the dynamic of the total implied variance is prescribed, 
whereas we derive it from the shape of the smile (SSVI in our case).
Also we identify the terminal condition on the total variance at maturity as a key property, and prove that there is no process satisfying this condition in our case.
Indeed our main result is at first disappointingly negative: 
starting from an uncorrelated SSVI smile at all times, we show that no It\^o process,
beyond the Black-Scholes model with time-dependent volatility, ensures that the option prices are martingales.
The by-products of this result, however, are interesting and informative.
We obtain explicitly joint dynamics for the underlying and the option prices such that, 
locally in time (that is until  some time before the true maturity of the option), 
the underlying price is a martingale, and so are all  the vanilla option prices, despite the fact that
the option prices are not given by the expectation of the final payoffs under an equivalent martingale measure. 
This implies that until this horizon, it is not possible to synthesise an arbitrage. 
Yet, this cannot last until the maturity of the option, and the market will then change regime. 
This naturally gives rise to the new concept of implied volatility bubbles.
We believe this intermediate regime (between the traditional arbitrage-free situation with specified dynamics until maturity and a regime with instantaneous arbitrages) is of interest,
and may correspond to real world situations.

We introduce precisely the SSVI parameterisation in Section~\ref{sec:SSVI},
and recall the notions of absence of arbitrage for a given implied volatility surface.
In Section~\ref{sec:SSVIDyn}, 
we introduce a new stochastic model describing the dynamics of the implied volatility surface, 
and extend the static arbitrage concept to a dynamic version.
We show there that unfortunately there cannot be any It\^o process solution in our setting.
However, this leads us to introduce implied volatility bubbles in Section~\ref{sec:SmileBubbles},
which we study in detail in the SSVI case.

\section{Static arbitrage-free volatility surfaces}\label{sec:SSVI}
We recall in this section the key ingredients of volatility surface parameterisation
as well as the different concepts of no (static) arbitrage in this setting.
This will serve as the basis of our analysis, and allows us to define properly the notion of dynamic arbitrage,
and its relation with martingale concepts.
In order to set the notations, recall that, in the Black-Scholes model~\cite{Black} 
with volatility $\sigma>0$, 
the price of a Call option with strike~$K>0$ and maturity $T>0$ is given at time~$t \in [0,T]$ by
\begin{equation}\label{eq:CallPrice}
\BSC(S_t, K, T, \sigma)
  = \EE\left[(S_{T}-K)_+\vert \Ff_t\right]
 = S_t\BS\left(\log\left(\frac{K}{S_t}\right),\sigma\sqrt{T-t}\right),
\end{equation}
for any $t \in [0,T]$, where the function $\BS:\RR\times\RR_+\to\RR$ is defined as
\begin{equation}\label{eq:DefBS}
\BS(k,v) :=
\left\{
\begin{array}{ll}
\displaystyle
\Nn\left(d_{+}(k,v)\right) - \E^{k}\Nn\left(d_{-}(k,v)\right), 
& \text{if } v>0,\\
\left(1-\E^{k}\right)_+, & \text{if } v=0,\\
\end{array}
\right.
\end{equation}
with~$\Nn$ the Gaussian cumulative distribution function, 
and $d_{\pm}(u,v) := \frac{-u}{v} \pm \frac{v}{2}$.
Practitioners generally do not work with option prices directly,
but rather with the Black-Scholes implied volatility map~$\sigma_t:\RR_+\times\RR_+\to\RR_+$ defined through the implicit relationship
$C_t^{\mathrm{obs}}(K,T)  = S_t\BS\left(k, \sigma_t(k,T)\sqrt{T-t}\right)$,
where~$C_t^{\mathrm{obs}}(K,T)$ denotes the observed option price at time~$t$ for a given strike and maturity,
and $k:=\log(K/S_t)$ is the log-moneyness.
There exist different conventions to write the implied volatility;
for reasons that will become apparent later, we choose to write it as a function of~$k$.
For convenience, we shall in fact work in terms of the total variance 
$\omega_t(k,T) := (T-t)\sigma_t^2(k,T)$, so that the Call price formula~\eqref{eq:CallPrice}
 can be rewritten, for any $t\in [0,T]$, as
$$
C_t^{\mathrm{obs}}(K,T)  = S_t\BS\left(k, \sqrt{\omega_t(k,T)}\right).
$$

\subsection{Static arbitrage}
Following~\cite{SSVIGat}, we recall the notion of static arbitrage, when for fixed running time~$t$, 
we only trade at~$t$ for a final maturity~$T$, but not in between.
\begin{definition}\label{def:StaticArb}
For any fixed $t\in [0,T]$, 
\begin{itemize}
\item the surface~$\omega_t(\cdot, \cdot)$ is free of calendar spread arbitrage if
$T\mapsto \omega_t(k, T)$ is increasing, for any $k\in\RR$;
\item a slice $k\mapsto \omega_t(k,T)$ is free of butterfly arbitrage if the corresponding density if non-negative.
\end{itemize}
A surface is free of static arbitrage if it is free of both calendar and butterfly arbitrages. 
\end{definition}
\noindent As hinted by its very name, this notion of arbitrage is static, 
in that it only concerns the marginal distributions of the stock price between~$t$ and~$T$, viewed at time $t$, but does not involve any dynamic behaviour in the running time~$t$.
It is equivalent to the impossibility of locking an arbitrage by trading in the option and the stock~$t$
and at the expiry of the option.
Before discussing a dynamic version of no-arbitrage, let us recall the SVI parameterisation,
a standard on Equity markets, 
which constitutes the backbone of our analysis.

\subsection{SSVI parameterisation}
Finding a parametric tractable model for a volatility smile has long been a challenge, 
and a breakthrough came when Gatheral~\cite{GatheralSVI} disclosed the SVI parameterisation
\begin{equation}\label{eq:SVI}
\omega(k,T) = \SVI(k) := a + b \left(\rho \left(k-m\right) + \sqrt{(k-m)^2 + \sigma^2} \right)
\end{equation}
for the total variance.
Since we only consider for now a fixed time~$t$, we drop the dependence thereof in the notation without confusion.
Here $a,b,\rho,m$ and $\sigma$ are parameters. 
This parameterisation only provides a characterisation of slices, 
so that the parameters are in principle different for each maturity~$T$.
The fit to market data is fairly good, and we refer the reader to~\cite{DeMarco} for an efficient and robust
dimension reduction calibration method. Necessary and sufficient conditions on the parameters preventing static arbitrage (Definition~\ref{def:StaticArb}) 
have been recently characterized in~\cite{arianna1st}.
To take into account the maturity dimension (hence the whole volatility surface, still without dynamics), Gatheral and Jacquier~\cite{SSVIGat} extended~\eqref{eq:SVI} to the Surface SVI (SSVI) parameterisation
\begin{equation}\label{eq:SSVI}
\omega(k, \theta(T)) := \frac{\theta(T)}{2} \left( 1 + \rho \varphi(\theta(T)) k + \sqrt{ (\varphi(\theta(T))k + \rho)^2 + \rrho^2} \right),
\end{equation}
where $T\mapsto\theta(T)$ is a non-decreasing and strictly positive function representing the at-the-money total implied variance, $\rho \in (-1, 1)$, $\rrho:=\sqrt{1-\rho^2}$, 
and~$\varphi$ is a smooth function from~$\RR_+^*$ to~$\RR_+^*$.
This formulation in fact enables one to find sufficient conditions to ensure absence of static arbitrage:

\begin{proposition}\label{prop:SSVINoArb}[Theorems 4.1 and 4.2 in~\cite{SSVIGat}] \ 
\begin{itemize}
\item There is no calendar spread  if
\begin{equation*}
\left\{
\begin{array}{ll}
\displaystyle \partial_t \theta(t) \geq 0,  & \text{for all } t>0,\\
\displaystyle 0 \leq \partial_{\theta} (\theta \varphi(\theta)) \leq \frac{1 + \rrho}{\rho^2} \varphi(\theta),
 & \text{for all } \theta>0;
\end{array}
\right.
\end{equation*}
\item There is no butterfly arbitrage if
$\displaystyle
\theta \varphi(\theta) \leq \min \left( \frac{4}{1+|\rho|}, 2 \sqrt{\frac{\theta}{1+|\rho|}} \right)$
for all $\theta>0$.
\end{itemize}
\end{proposition}

\subsubsection{Symmetric SSVI}
Setting $\rho=0$ leads to a symmetric smile in log-moneyness.
Since we consider here a single smile, we only investigate Butterfly arbitrage.
In the uncorrelated case $\rho=0$, the no-Butterfly arbitrage condition in Proposition~\ref{prop:SSVINoArb} can be simplified to
$\displaystyle 
\theta \varphi(\theta) \leq \min (4, 2\sqrt{\theta})$, 
for all $\theta>0$.
In this case, a slight improvement,
as an explicit necessary and sufficient formulation, was provided in~\cite{SSVIGat}.
Define 
\begin{equation}\label{eq:BigB}
\B(\theta) := A(\theta)\ind_{\{\theta < 4\}}+16\,\ind_{\{\theta \geq 4\}}
\end{equation}
from~$\RR_+$ to~$\RR_+$, where, for any $\theta>0$, 
$$
A(\theta) =  \frac{16 \theta \zeta_{\theta}\left(\zeta_{\theta}+1\right)}{8\left(\zeta_{\theta}-2\right) + \theta \zeta_{\theta}\left(\zeta_{\theta}-1\right)}, 
\qquad\text{with}\qquad
\zeta_{\theta} := \frac{2}{1-\theta/4} + \sqrt{\left(\frac{2}{1-\theta/4}\right)^2 + \frac{2}{1-\theta/4}}.
$$
\begin{corollary}\label{cor:ButterflyArb}
If $\rho = 0$, there is no Butterfly arbitrage if and only if
$(\theta\varphi(\theta))^2 \leq \B(\theta)$ for all $\theta>0$.
\end{corollary}
Since
$\lim_{\theta \downarrow 0} \sqrt{A(\theta)/\theta}  = \cs$,
with $\cs\approx 4.45$, we have approximately a gain of a factor two with respect to the simplified sufficient condition.
We note in passing that extended versions of SSVI have since been developed~\cite{corbetta2019robust, SSVIGen, eSSVI}, 
where again sufficient conditions are provided to ensure absence of static arbitrage.

\section{Dynamic arbitrage-free volatility surfaces}\label{sec:SSVIDyn}
Static arbitrage is by now well understood and has contributed 
to providing valid examples to generate market options data and to design interpolators and extrapolators of option quotes.
In the static setting above, no dynamics was set for any of the ingredients.
We now extend this framework to a dynamic setting, where both the stock price and the implied volatility evolve.
We fix a filtered probability space $(\Omega, \Ff, (\Ff_t), \QQ)$ 
on which all processes and Brownian motions are well defined,
and consider a stochastic stock price adapted to~$(\Ff_t)$.
The main novelty of our approach, which is common for the moment with  the early works by Carr and Sun~\cite{carr2014implied} and Carr and Wu~\cite{CarrWu}, is to impose some dynamics
for the total implied variance $(\omega_t(k, T))_{t\in [0, T]}$.
The maturity~$T>0$ is fixed throughout the paper, and hence the implied volatility at time~$t$ 
only makes sense for $t\in [0, T)$.
We first introduce the concept of dynamic arbitrage without specifying any dynamics.

\subsection{Dynamic arbitrage and consistent total variance models.}
One key point is that we now write $k_t:=\log(K/S_t)$
instead of~$k$ for the log-moneyness, emphasising the importance of the running time~$t$.

\begin{definition}\label{def:Consistent}
A consistent total variance model is a couple $(S_t, \omega_t(k_t, T))_{t \in [0, T], K>0}$ such that, up to~$T$,
\begin{enumerate}[(i)]
\item the process $S$ is a strictly positive $\QQ$-martingale with continuous sample paths;
\item for every $K>0$, the process $\omega_{\cdot}(k_{\cdot},T)$ has continuous paths and is 
strictly positive on $[0,T)$;
\item for every $K>0$, $\omega_t(k_t,T)$ converges to zero almost surely as~$t$ approaches~$T$;
\item for every $K>0$, the process~$C$ defined by 
$C_t :=S_t\BS\left(k_t, \sqrt{\omega_t(k_t,T)}\right)$ is a $\QQ$-martingale.
\end{enumerate}
We denote~$\Qq_T$ the set of all consistent total variance models,
and no dynamic arbitrage occurs if~$\Qq_T\ne\emptyset$.
\end{definition}

By Put-Call-Parity we can equivalently replace the last item above by the martingale property of the Put price process, directly through the Black-Scholes Put pricing function. 
The following useful remark relaxes Condition~(iv) above from, replacing it effectively by a local martingale assumption:
\begin{lemma}\label{truemart}
Let $K>0$ and assume that  the process $(C_t)_{t\in [0,T]}$ defined by 
$C_t :=S_t\BS\left(k_t, \sqrt{\omega_t(k_t,T)}\right)$ is a $\QQ$-local martingale.
Then if~$S$ is a martingale, so is~$C$.
\end{lemma}

\begin{proof}
Indeed the process~$P$ defined as $P_t:=C_t -(S_t-K)$ is a local martingale which is positive and uniformly bounded by $K$, hence a martingale, and therefore $C = P + (S-K)$ is a martingale as well.
\end{proof}

At first glance, there is no link between the dynamics of each option contract (indexed by~$K$), 
so that the option could evolve in an inconsistent manner even when starting from a static arbitrage-free configuration; 
this is actually not the case due to Definition~\ref{def:Consistent}(iii):
\begin{lemma} \label{DynamicImpliesStatic}
Absence of dynamic arbitrage implies absence of Butterfly arbitrage.
\end{lemma}

\begin{proof}
We claim that the Call price is the conditional expectation of the payoff $(S_T-K)_+$.
Consider a Put option with price~$P_t := \BSP(S_t, K, \omega_t(k_t,T))$, 
where~$\BSP$ denotes the Black-Scholes Put option price. 
In absence of dynamic arbitrage, Definition~\ref{def:Consistent}
implies that
$P_t = \EE_t^{\QQ}[P_T]$ and 
$P_T = \BSP(S_T, K, \omega_T(k_T,T)) = \BSP(S_T, K, 0) = (K-S_T)_+$ almost surely. 
Since the payoff~$P_T$ is uniformly bounded by~$K$,  
then dominated convergence implies that $P_t = \EE_t[(K-S_T)_+]$.
Since 
$C_t(K,T) = \BSC(S_t, K, \omega_t(k_t,T))
 = S_t-K - \BSP(S_t, K, \omega_t(k_t,T)) =  S_t-K - \EE_t[(K-S_T)_+]
  = \EE_t[(S_T-K)_+]$, 
  the claim follows.
\end{proof}
Definition~\ref{def:Consistent}(iii) is not present in~\cite{carr2014implied} where nothing prevents inconsistent situations to occur.
If Definition~\ref{def:Consistent}-(i)-(iv) hold, each individual price being a martingale, 
no arbitrage can be exploited from trading individually in the options or stocks. 
Yet Butterfly arbitrage (even in its simple form of the non-monotonicity of the Call option price with respect to the strike) could occur and be exploited. 
As an example, consider two Call options with maturity~$T$, 
strikes~$K_1$ and~$K_2$, with $0<K_1<K_2$, and with dynamics given by 
$$
C_t(K_1,T) = C_0(K_1) \exp\left\{\sigma_1 B^{(1)}_{t} - \frac{1}{2} \sigma_1^2 t \right\}
\qquad\text{and}\qquad
C_t(K_2,T) = C_0(K_2) \exp\left\{\sigma_2 B^{(2)}_{t} - \frac{1}{2} \sigma_2^2 t \right\},
$$
for $t\in [0,T]$, given two independent Brownian motions~$B^{(1)}$ and~$B^{(2)}$. 
Assume further that~$S$ follows yet another Black-Scholes-type dynamics 
$S_t = S_0 \exp\{\sigma_2 \overline{B}_{t} - \frac{1}{2} \sigma_2^2 t \}$
where~$\overline{B}$ is a Brownian motion independent of~$B^{(1)}$ and~$B^{(2)}$, 
such that $(S_0-K_2)_+ < C_0(K_2) < C_0(K_1)<S_0$.
For $i=1,2$ introduce the exit times 
$$
\tau_i := \inf\left\{t\in[0,T): C_t(K_i) \notin \left((S_t - K_i)_+, S_t\right)\right\},
$$
and, for any $\eps>0$, the crossing time
$\widetilde{\tau}_\eps := \inf\{ t: C_t(K_2)>C_t(K_1)+\eps \}$. 
Then $\{\widetilde{\tau}_\eps < \tau_1 \wedge \tau_2 \}$ has positive probability, 
and prices become inconsistent at~$\widetilde{\tau}_\eps$, although each individual option process is a martingale.

\subsection{The  dynamic symmetric SSVI}
In order to be more precise, we now specify some dynamics:
\begin{equation}\label{eq:SDES}
\D S_t = S_t\sqrt{v_t}\,\D B^S_t,
\end{equation}
starting without loss of generality from $S_0=1$, for some Brownian motion~$B^S$.
Here the process~$(v_t)_{t\geq 0}$ is left unspecified, 
but regular enough (and non-negative) so that~\eqref{eq:SDES} admits a unique weak solution.
Since we want the process~$S$ to be a true martingale, we impose the Novikov condition 
\begin{equation} \label{Novikov}
\EE\left[\exp\left\{\frac{1}{2}\int_0^T v_t \D t\right\}\right] < \infty.
\end{equation}
We consider a dynamic version of the uncorrelated SSVI parameterisation~\eqref{eq:SSVI}, namely
\begin{equation}\label{eq:SSVIDyn}
\omega_t(k_t, \theta_t)= \frac{\tht}{2}\left(1+\sqrt{1+\varphi^2_t k_t^2}\right).
\end{equation}
where, similar to~\eqref{eq:SSVI}, $\tht$ accounts now for the at-the-money total implied variance
at~$t$ for an option maturing at~$T$.
We implicitly disentangled here the link between the function~$\varphi_t$
and the curve $T\mapsto\theta_t$, by introducing the process notation~$(\varphi_t)$.
We keep this terminology from now on, and reverting back to the classical SSVI~\eqref{eq:SSVI} boils down to a simple change of variables.
We further assume that~$\theta$ and~$\varphi$ are diffusion processes given by
\begin{equation}\label{eq:thetaphiDyn}
\left\{\begin{array}{rll}
\D\theta_t & = \displaystyle \theta_{1,t} \D t + \theta_{2,t}\D B_t^{\theta}, & \theta_0 >0,\\
\D\varphi_t & = \displaystyle \varphi_{1,t} \D t + \varphi_{2,t}\D B_t^{\varphi}, & \varphi_0>0,\\
\D\la B^{\theta}, B^{\varphi}\ra_t & = \displaystyle \varrho \, \D t,
\end{array}
\right.
\end{equation}
where $B^{\theta}$ and $B^{\varphi}$ are two Brownian motions.
The time-dependent coefficients $\theta_1, \theta_2, \varphi_1, \varphi_2$ are left unspecified,
and may be stochastic, adapted to the filtration $(\Ff_t)_{t\in [0,T]}$ and such that the two stochastic diffusions admit unique weak solutions.
Therefore

\begin{lemma}\label{le:thetazero}
A necessary condition for Definition~\ref{def:Consistent}(iii) is that $\theta$ converges to zero
amost surely at time~$T$.
\end{lemma}

We investigate here the existence of consistent total variance models 
of the form~\eqref{eq:SDES}-\eqref{eq:SSVIDyn}-\eqref{eq:thetaphiDyn}.
Given the symmetry of the implied volatility~\eqref{eq:SSVIDyn}, 
we expect~$\theta$ and~$\varphi$ to depend solely on~$v$ and on its driving Brownian, 
assumed independent of~$B^S$.
The maturity~$T$ does not come into play explicitly in our parameterisation, 
but is present through~$\theta_t$ and~$\varphi_t$. 
The martingale~$S$ induces the new measure $\D\QQ = S_T \D\PP$,
and Girsanov's theorem implies that the process~$\widetilde{W}$ defined by
$\D\Wt_t := \D B^S_t -\sqrt{v_t} \D t$ is a $\QQ$-Brownian motion. 
Hence for any given $(K,T)$, a Call option~$C_{\cdot}(K,T)$ is a $\PP$-martingale if and only if
the process~$\Ct_{\cdot}(K,T)$, defined as 
$\Ct_t(K,T):=C_{t}(K,T)/S_t = \bs(k_t,\sqrt{\omega_t(k_t,T)})$,
is a $\QQ$-martingale. 
The terminal condition on the Call prices under~$\PP$ is that
$S_t \bs(k_t,\sqrt{\omega_t(k_t,T)})$ converges to
the intrinsic payoff $(S_T-K)_+$ almost surely as~$t$ tends to~$T$, 
which is granted as soon as $\omega_t(k_t,T)$ converges to zero almost surely.
The objective is to find conditions on the parameters~$\theta_{1,t}$, $
\theta_{2,t}$, $\varphi_{1,t}$, $\varphi_{2,t}$ and~$\varrho$ in~\eqref{eq:thetaphiDyn} ensuring no dynamic arbitrage.

\begin{theorem}\label{thm:thetaPhiDynamic}
If there is a consistent total variance model, in the sense of Definition~\ref{def:Consistent}, 
then necessarily,
\begin{equation}\label{eq:SolutionSystem}
\left\{
\begin{array}{rl}
\D\theta_t & = \displaystyle \frac{(\theta_t\varphi_t-4)(\theta_t\varphi_t+4)}{16}v_t\D t - \theta_t\varphi_t\sqrt{v_t}\,\D B_t,\\
\D\varphi_t & = \displaystyle \left(16+16\varphi_t^2\theta_t-\theta_t^2\varphi_t^2\right)\frac{\varphi_t v_t}{16\theta_t}\D t+\varphi_t^2\sqrt{v_t}\,\D B_t,
\end{array}
\right.
\end{equation}
where~$B$ is a Brownian motion independent from~$B^S$. 
\end{theorem}
We stress that the statement is only necessary. 
Nothing grants the existence of an actual solution; 
moreover, even if one exists, the following three conditions should be checked 
for the solution to be valid:
\begin{itemize}
\item both processes $\theta$ and $\varphi$ should be positive almost surely;
\item the no-arbitrage Condition~\ref{cor:ButterflyArb} should hold;
\item the boundary condition $\theta_T$ should be null almost surely;
\end{itemize}

Lemma~\ref{lem:Existence} below in fact shows that existence is a real issue.
An important remark here is that the Brownian motion~$B$ may not be related to the dynamics of the variance process~$(v_t)$, 
as the latter does not come into play at any stage in the computations. 
We will discuss this more in detail in Section~\ref{sec:SmileBubbles} below.

\begin{proof}[Proof of Theorem~\ref{thm:thetaPhiDynamic}]
To simplify the computations, introduce the notations
$$
Y_t := \varphi_t k_t,\qquad
\eta_t : = h(\theta_t),\qquad
\gamma_t := f(Y_t),\qquad
\Omega_t := \gamma_t \eta_t,\qquad
f(y):=\sqrt{1+\sqrt{1+y^2}}, \qquad
h(\theta):=\sqrt{\theta/2}.
$$
This implies that, in the symmetric SSVI framework~\eqref{eq:SSVIDyn}, 
the Call price function~\eqref{eq:CallPrice} simplifies to
$$
C_t^{\mathrm{obs}}(K,T) = S_t \bs\left(k_t,\Omega_t\right) =: S_t \Ct_t(k_t, \Omega_t).
$$
It\^o's formula implies that, for any (fixed) $k\in\RR$ and any $t>0$, we can write
$$
\D\Ct_t = \partial_1 \bs(k_t, \Omega_t) \D k_t + \partial_2 \bs(k_t, \Omega_t) \D\Omega_t + \frac{1}{2} \partial_{11}^2 \bs(k_t, \Omega_t) \D\la k\ra_t + \frac{1}{2} \partial_{22}^2 \bs(k_t, \Omega_t) \D\la\Omega\ra_t + \partial_{12}^2 \bs(k_t, \Omega_t) \D\la k, \Omega\ra_t.
$$
The derivatives of the $\bs(\cdot, \cdot)$ function are classical and straightforward:
\begin{equation*}
\begin{array}{rlrl}
\displaystyle \partial_1 \bs(u,v) & = \displaystyle -\E^{u} \Nn(\dm),
\qquad & \displaystyle \partial_2 \bs(u,v) & = \displaystyle n(\dpp),\\
\displaystyle \partial_{11}^2 \bs(u,v) & = \displaystyle -\E^{u} \Nn(\dm)+\frac{n(\dpp)}{v},
\qquad & \displaystyle \partial_{12}^2 \bs(u,v) & = \displaystyle \left(\frac{1}{2}-\frac{u}{v^2}\right) n(\dpp),\\
\displaystyle \partial_{22}^2 \bs(u,v) & = \displaystyle \left(\frac{u^2}{v^3} -\frac{v}{4}\right) n(\dpp).
\end{array}
\end{equation*}
Now, we can write the dynamics for all the processes appearing in this equation as
\begin{equation*}
\left.
\begin{array}{rlrl}
\D k_t & = \displaystyle \sqrt{v_t} \D\Wt_t -\frac{v_t}{2}\D t
 & \D\la k\ra_t & = \displaystyle v_t \D t,\\
\D Y_t & \displaystyle = \varphi_t \D k_t + \frac{Y_t}{\varphi_t} \D\varphi_t + \D\la \varphi, k\ra_t,
 & \D\la Y\ra_t & \displaystyle = \varphi_t^2 \D\la k\ra_t + \left(\frac{Y_t}{\varphi_t}\right)^2 \D\la \varphi\ra_t + 2 Y_t \D\la \varphi, k\ra_t,\\
\D\gamma_t & = \displaystyle f'(Y_t)\D Y_t + \frac{1}{2} f''(Y_t)\D\la Y\ra_t,
 & \D\la\gamma\ra_t & = \displaystyle f'(Y_t)^2 \D\la Y\ra_t,\\
\D\Omega_t & = \displaystyle \eta_t \D\gamma_t + \gamma_t\D\eta_t + \D\la\eta, \gamma\ra_t,
 & \D\la\Omega\ra_t & = \eta_t^2 \D\la\gamma\ra_t + \gamma_t^2 \D\la\eta\ra_t + 2 \eta_t \gamma_t \D\la\eta,\gamma\ra_t,\\
\D\eta_t & \displaystyle = h'(\theta_t)\D\theta_t + \frac{1}{2} h''(\theta_t)\D\la\theta\ra_t,
 & \D\la\eta\ra_t & \displaystyle = h'(\theta_t)^2 \D\la \theta\ra_t,\\
\D\la k,\Omega\ra_t & \displaystyle = \eta_t f'(Y_t)\left(\varphi_t \D\la k\ra_t
 + \frac{Y_t}{\varphi_t} \D\la k,\varphi\ra_t\right) + \gamma_t \D\la k,\eta\ra_t,
  &  \D\la\eta, \gamma\ra_t
 & = \displaystyle h'(\theta_t) f'(Y_t) \left(\varphi_t \D\la \theta, k\ra_t
   + \frac{Y_t}{\varphi_t} \D\la \theta, \varphi\ra_t\right)\\
\D\la k, \eta\ra_t & = h'(\theta_t)\D\la k, \theta\ra_t.
\end{array}
\right.
\end{equation*}
In order for the option price to be a martingale, the drift part~$\D\Ct$ has to be equal to zero.
Assume now, as in the proposition, that
$\la\theta, \varphi\ra_t = \psi_t = \la\theta, k\ra_t = 0$.
Long and tedious computations (that we perform with the counter checks of Sympy) show that the latter can be written as 
$Z_t^2(1+Z_t)^2\Pp(Z_t)\D t$, where~$\Pp(\cdot)$ is a fifth-order polynomial, 
where $Z_t:=\sqrt{1+Y_t^2}$.
In particular, each coefficient~$\Pp_{i}$ of order $i=0,\ldots, 5$ read
\begin{equation*}
\begin{array}{rcl}
\Pp_5 & = & 
\left(\pf_t - 16\right)\left(\varphi_t^2\theta_{2,t}^2+\theta_t^2\varphi_{2,t}^2+2\sqrt{\pf_t}\chi_t\right),\\
\Pp_4 & = & -2\theta_t\Big\{
\left(8\varphi_t\sqrt{\pf_t} - \pf_t -16\right)\varphi_t\chi_t
-\varphi_t^4 \theta_t \theta_{2,t}^2
+16\varphi_t^4 \theta_t \theta_{1,t}
-4\varphi_t^4 \theta_{2,t}^2
+16\varphi_t \pf_t\varphi_{1,t}
-4\pf_t\varphi_{2,t}^2
-16\theta_t \varphi_{2,t}^2
\Big\},\\
\Pp_3 & = & \varphi_t\Big\{2\left(16-\pf_t\right)\theta_t\chi_t
 + \varphi_t\left[
  \pf_t\left(\pf_t  - 16\right) v_t + \pf_t\theta_{2,t}^2  - 2\theta_t^4\varphi_{2,t}^2
 - 32\pf_t\theta_{1,t}
 + 8\varphi_t\sqrt{\pf_t}\theta_{2,t}^2
 - 8\theta_t^3\varphi_{2,t}^2 + 16\theta_{2,t}^2\right]
\Big\},\\
\Pp_2 & = &  2\theta_t\Big\{
(8\varphi_t\sqrt{\pf_t} - \pf_t-16)\varphi_t\chi_t
+4\theta_t\left(\varphi_t^6\theta_t v - 4\varphi_t^4 v + 4\varphi_t^3\theta_t\varphi_{1,t} - 3\varphi_t^2\theta_t\varphi_{2,t}^2 - 4\varphi_{2,t}^2\right)
\Big\},\\
\Pp_1 & = & 
-\theta_t^2\left(\pf_t+8\varphi_t\sqrt{\pf_t}+16\right)\left(\varphi_t^4 v_t - \varphi_{2,t}^2\right),\\
\Pp_0 & = & -16\pf_t\theta_t\left(\varphi_t^4 v - \varphi_{2,t}^2\right),\\
\end{array}
\end{equation*}
where we introduced $\pf_t:=\varphi_t^2\theta_t^2$,
and write~$\chi_t$ for the drift term of the covariation $\D\la\theta, \varphi\ra_t$.

Now, the dependence on~$K$ and the running spot~$S_t$ in the drift of~$\D\Ct$
is only through~$Z_t$, as both~$\theta_t$ and~$\varphi_t$ are independent thereof.
The derivation above is valid for any fixed~$K$, thus for all~$K>0$.
The only way the drift condition can be achieved is therefore that $\Pp(Z_t)$ is identically null, 
which holds if and only if $\Pp_i=0$ for $i=0,\ldots, 5$.
This system, with unknown $(\theta_{1,t}, \theta_{2,t}, \varphi_{1,t}, \varphi_{2,t}, \chi_t)$, 
is solvable as $\mathfrak{S}_+$ and~$\mathfrak{S}_-$, where
$$
\mathfrak{S}_{\pm} = 
\begin{pmatrix}
\displaystyle \frac{v_t}{16}\left(
\left[2\varrho^2-1\right]\varphi_t^2\theta_t^2
 + 8\varphi_t^2\theta_t\left[\varrho^2-1\right]
  - 2\varphi_t^2\theta_t\left[\theta_t+4 \right]\overline{\varrho}|\varrho|
   -16\right)\\
\displaystyle \pm \varphi_t\theta_t\sqrt{v_t}\left(|\varrho|\overline{\varrho}/\varrho - \varrho\right)\\
\displaystyle \frac{v_t\varphi_t}{16\theta_t}\left(
8\varphi_t^2\theta_t\left[1+\varrho^2\right]
 - \varrho^2\varphi_t^2\theta_t^2
 - 16\varrho^2
   + \left[16 - 8\varphi_t^2\theta_t + \varphi_t^2\theta_t^2\right]\overline{\varrho}|\varrho|
     + 32\right)\\
\mp\varphi_t^2\sqrt{v_t}\\
\varphi_t^3\theta_t v_t\left(\overline{\varrho}|\varrho| - \varrho^2\right)
\end{pmatrix},
$$
and $\overline{\varrho}:=\sqrt{1-\varrho^2}$.
It admits a solution only for 
$\varrho \in \{-1,1\}$, 
and~\eqref{eq:SolutionSystem} follows.
The two solutions for~$\theta_{2,t}$ and~$\varphi_{2,t}$ are hence of opposite sign, 
hence equivalent as Brownian increments are symmetric.

We now show that the absence of correlation between~$S$ and both~$\varphi$ and~$\theta$ is in fact necessary.
Denote by~$\D\Ct_{+}$ the expression above for~$\D\Ct$, and let $\D\Ct_{-}$ be the same one, 
except that we now set $Y_t:=-\sqrt{Z_t^2-1}$.
Since the smile is symmetric, then the difference $\D(\Ct_{+}-\Ct_{-})$ has to be null.
This yields the following system:
\begin{equation*}
\left\{
\begin{array}{rl}
\theta_{t}^{2} \varphi_{t}^{2}\D\la\varphi,k\ra_t & = 0,\\
\left(\theta_{t}^{2} \varphi_{t}^{2} + 8 \theta_{t} \varphi_{t}^{2} + 16\right)\theta_t \D\la\varphi,k\ra_t & = 0,\\
\varphi_{t} \Big(\D\la\theta,k\ra_t \theta_{t}^{2} \varphi_{t}^{2} - 8 \D\la\theta,k\ra_t \theta_{t} \varphi_{t}^{2} + 16 \D\la\theta,k\ra_t - 8 \D\la\varphi,k\ra_t \theta_{t}^{2} \varphi_{t}\Big) & = 0,\\
\Big(\D\la\theta,k\ra_t \varphi_{t} + \D\la\varphi,k\ra_t \theta_{t}\Big) \Big(\theta_{t}^{2} \varphi_{t}^{2} + 16\Big) & = 0,
\end{array}
\right.
\end{equation*}
and it is easy to see that the only valid solution is
$\D\la\theta, k\ra_t = \D\la\varphi, k\ra_t = 0$.
\end{proof}

\begin{proposition}\label{prop:CVTFlat}
If $\Qq_T\ne \emptyset$ and $(\theta, \varphi)$ solves~\eqref{eq:SolutionSystem},
then for any $t\in [0,T]$, 
$\theta_t = \int^T_t v_u \D u$, all the smiles (at~$t$, maturing at~$T$) are flat, 
and $\int^T_t v_u \D u \in \Ff_t$.
\end{proposition}
\begin{proof}
The condition $\theta_T=0$ implies $\psi_T=0$. 
The SSVI smile is then trivial, equal to~$\theta_t$ for all strikes.
Furthermore $\D\theta_t=-v_t \D t$, and hence $\theta_t = \int^T_t v_u \D u$.
Since $\theta_t$ is $\Ff_t$ measurable, then so is $\int^T_t v_u \D u$.
\end{proof}
The Black-Scholes model with deterministic volatility clearly satisfies these conditions, 
and it is immediate to check from the definition that it is indeed a consistent total variance model.

\begin{remark}
If the property $\int^T_t v_u \D u \in \Ff_t$ is taken to hold for every $T>0$, then $v_T$ is $\Ff_t$-measurable, 
and~$v$ deterministic. 
Note that we only considered here a fixed maturity~$T$, not all of them.
One could then wonder whether the necessary properties above imply that~$v$ is deterministic, 
in which case the only solution would be a time-dependent Black-Scholes model.
We provide a non-trivial example showing that this is not always the case:
on the filtration of a planar Brownian motion $(W,B)$, fix $T=1$, $0<\eps_1, \eps_2<1$, 
and define
\begin{equation*}
v_t := 
\left\{
\begin{array}{ll}
v_0, & \text{for } \displaystyle t \in [0, 1/3),\\
\displaystyle v_0(1+\eps_1)\ind_{B_{1/3}>0} +  v_0(1-\eps_2)\ind_{B_{1/3} \leq 0},
 & \text{for } \displaystyle t\in [1/3, 2/3),\\
\displaystyle v_0(1-\eps_1)\ind_{B_{1/3}>0} +  v_0(1+\eps_2)\ind_{B_{1/3} \leq 0},
 & \text{for } \displaystyle t\in [2/3, 1).
 \end{array}
\right.
\end{equation*}
Then clearly $ \int^T_t v_u \D u$ is $\Ff_t$-measurable for every $t\in [0,1]$, 
and therefore the conditional law of~$S_T$ is lognormal, and the smiles are flat at every point in time. 
This in particular yields another example of a consistent total variance model, different from the time-dependent Black-Scholes.
\end{remark}

\subsection{Non-existence of non-trivial consistent dynamics}
We now give three different proofs that there is no non-trivial  consistent dynamics; 
they all rely on the core observation that~$(\tht)_{t\in[0,T]}$ should tend to zero at maturity as shown in Lemma~\ref{le:thetazero}.
Notice from~\eqref{eq:SolutionSystem} that $\D (\theta_t\varphi_t) = 0$,
so that $\theta_t \varphi_t = \psi_T$ for any $t\in [0,T]$ for some constant~$\psi_T$.
Hence, the SDE for~$\theta$ can be rewritten more compactly as
$$
\D \theta_t = \frac{(\psi_T-4)(\psi_T+4)}{16}v_t \D t - \psi_T\sqrt{v_t}\,\D B_t,
\qquad\text{with boundary condition }\theta_T=0.
$$

\subsubsection{First proof}
With the necessary and sufficient no-Butterfly arbitrage condition $(\theta\varphi(\theta))^2 \leq \B(\theta)$,
given that $\theta_t\varphi(\theta_t)=\psi_T$, and that~$\B$ tends to zero when~$\theta$ does,
there cannot be a consistent dynamic;
otherwise there would be no static arbitrage at any point in time (Lemma~\ref{DynamicImpliesStatic}),
and since necessarily~$\theta$ tends to zero as~$t$ approaches~$T$ (Lemma~\ref{le:thetazero}),
 $\B(\theta_t) <(\theta_t\varphi(\theta_t))^2$ and there is a Butterfly arbitrage.

\subsubsection{Second proof}\label{sec:Proof2}
We investigate whether the solution in Theorem~\ref{thm:thetaPhiDynamic} is trivial or not, without using the 
knowledge of the shape of the exact no-Butterfly arbitrage region.
 
\begin{lemma}\label{lem:Existence} 
Let $\alpha:=-\left(\frac{\psi_T}{16} - \frac{1}{\psi_T}\right)$. 
If $\EE\left[ \exp\left\{2 \alpha^2 \int_0^T v_t \D t\right\}\right]$ is finite 
(in particular, under the Novikov condition~\eqref{Novikov} if $\alpha<\frac{1}{2}$),
then any solution of~\eqref{eq:SolutionSystem} on $[0,T]$ satisfies
$$
\theta_t = \frac{\psi_T}{2\alpha}\log\EE_t\left[\exp\left(\frac{2\alpha \theta_T}{\psi_T}\right)\right],
\quad\text{almost surely for all } t \in [0, T].
$$
In particular, $\theta_T=0$ almost surely entails that  $\theta_t=0$ almost surely for all $t \in [0,T]$.
\end{lemma}
\begin{proof}
Setting $\widetilde{\theta}_t := -\frac{\theta_t}{\psi_T}$ yields
$\D \widetilde{\theta}_t = \alpha v_t \D t + \sqrt{v}_t \D B_t$,
with $\widetilde{\theta}_T = 0$ and $\alpha:= \frac{1}{\psi_T} -\frac{\psi_T}{16}$.
It\^o's formula implies
\begin{align*}
\E^{-2\alpha \widetilde{\theta}_T}
 & = \E^{-2\alpha \widetilde{\theta}_t} -2\alpha\int_{t}^{T}\E^{-2\alpha \theta_u}\D \widetilde{\theta}_u
 + 2\alpha^2\int_{t}^{T}\E^{-2\alpha \widetilde{\theta}_u}v_u \D u\\
 & = \E^{-2\alpha \widetilde{\theta}_t} -2\alpha\int_{t}^{T}\E^{-2\alpha \widetilde{\theta}_u}\left[\alpha v_u \D u + \sqrt{v}_u \D B_u\right]
  + 2\alpha^2\int_{t}^{T}\E^{-2\alpha \widetilde{\theta}_u}v_u \D u\\
 & =  \E^{-2\alpha \widetilde{\theta}_t} - 2\alpha\int_{t}^{T}\E^{-2\alpha \widetilde{\theta}_u}\sqrt{v}_u \D B_u.
\end{align*}
Note that $-2\alpha \widetilde{\theta}_u>0$ and $\E^{-2\alpha \widetilde{\theta}_u}=\E^{-2\alpha \widetilde{\theta}_0}\E^{-2\alpha^2 \int_0^u v_t \D t} \E^{-2\alpha \int_0^u \sqrt{v}_t \D B_t}$;
the Novikov condition applied to the process $-2\alpha \int_0^{\cdot} \sqrt{v}_t \D B_t$
reads $\EE\left[\exp\left\{2 \alpha^2 \int_0^T v_t \D t\right\}\right]<\infty$, 
in which case the stochastic integral is square integrable with null expectation, 
and every term in the expression has finite expectation. Taking expectations conditional on~$\Ff_t$ on both sides yields
$\EE_t\left[\E^{-2\alpha \widetilde{\theta}_T}\right] = \E^{-2\alpha \widetilde{\theta}_t}$
almost surely for all $t \in [0, T]$,
hence 
$$
\theta_t = \frac{\psi_T}{2\alpha}\log\EE_t\left[\exp\left(\frac{2\alpha \theta_T}{\psi_T}\right)\right],
\quad\text{almost surely for all } t \in [0, T].
$$
Imposing $\theta_T=0$ almost surely implies that $\theta_t = 0$ almost surely for all $t\in [0,T]$.
\end{proof}

\subsubsection{Third proof}\label{sec:Proof3}
Because of the terminal condition $\theta_T=0$, the stochastic differential equation satisfied by~$\theta$
is actually a classical example of a BSDE. 
Consider indeed the process
\begin{equation}\label{eq:BSDE}
X_t = \xi + \int_{t}^{T}f(s,Z_s)\D s - \int_{t}^{T}Z_s\D B_s,
\end{equation}
on $[0, T]$, with terminal condition $X_T = \xi$, 
where~$\xi$ is a bounded, $\Ff^B_T$-measurable random variable, 
and $f:[0,T]\times\RR\to \RR$ has at most quadratic growth.
This is the classical (one-dimensional) example of a quadratic BSDE~\cite[Chapter 10]{TouziBook}.
In our setting, integrating the SDE for~$\widetilde{\theta}$ on $[t,T]$ reads
$$
\widetilde{\theta}_t = \widetilde{\theta}_T
+ \int_{t}^{T}\left(\frac{\psi_T}{16} - \frac{1}{\psi_T}\right)v_s \D s - \int_{t}^{T}\sqrt{v}_s \D B_s,
$$
which is exactly of the form~\eqref{eq:BSDE} with 
$X = \widetilde{\theta}$,
$f(\cdot, z) \equiv (\frac{\psi_T}{16} - \frac{1}{\psi_T})z^2$, $Z = \sqrt{v}$ and $\xi = 0$ almost surely.
From~\cite[Chapter 10]{TouziBook}, if $(X,Z)$ is a solution to~\eqref{eq:BSDE} with~$X$ bounded, then $Z \in \HH^2_{\BMO}$, 
where 
$\HH^2_{\BMO} := \left\{\varphi \in \HH^2: \left\| \int_{0}^{\cdot}\varphi_s \D B_s\right\|_{\BMO}<\infty\right\}$,
with~$\HH^2$ the space of square integrable martingales, and
$\displaystyle 
\|M\|_{\BMO} := \sup\{\left\|\EE\left[\langle M\rangle_T - \langle M\rangle_{\tau}\vert\Ff_\tau\right]\right\|_{\infty}: \tau \in \Tt_{0}^{T}\},
$
where $\Tt_0^T$ denotes the set of stopping times in $[0,T]$.
Since the driver~$f$ is quadratic and smooth, uniqueness of a solution to~\eqref{eq:BSDE}
is guaranteed by~\cite[Theorem 10.5]{TouziBook}.
Because the terminal condition~$\xi$ is bounded here,
existence of a unique solution to~\eqref{eq:BSDE} is then given by~\cite[Theorem 10.6]{TouziBook},
with upper bound estimate for the (appropriate) norms of~$X$ and~$Z$. This grants that $(X, Z)=(0,0)$
is the unique solution.

\begin{remark}
Section~\ref{sec:Proof2} uses computations common in BSDE theory, 
and Section~\ref{sec:Proof3} is thus a shortcut for BSDE wizards;
so those proofs are essentially one and the same.
\end{remark}

\section{Smile bubbles}\label{sec:SmileBubbles}
We showed previously that, except for dynamics close to Black-Scholes, 
there is not dynamic model for~$(\theta_t)_{t\in[0,T]}$ and~$(\varphi_t)_{t\in[0,T]}$ 
of symmetric SSVI up to maturity~$T$. 
We now investigate what happens if we restrict ourselves to a shorter time horizon;
more precisely, we wish to find valid dynamics for~$\theta$ up to some (possibly stochastic) time horizon $\ts\in (0,T)$.  
This would constitute a dynamic arbitrage-free model on $[0, \ts]$, 
satisfying the martingale condition for the stock and Vanilla option price, 
with total variance given by SSVI. 
One caveat is that absence of Butterfly arbitrage is not granted, 
because the Vanilla prices can a priori no longer be written as conditional expectations of the terminal payoff, 
although the individual option prices are martingales.  
Now, in the case of symmetric SSVI, we have explicit necessary and sufficient
conditions on the parameters preventing Butterfly arbitrage: 
so starting in this domain, we will have locally no dynamic nor Butterfly arbitrages
up to the exit time~$\tau$ of this domain.
In order to state the definitions below, let~$\Tt_T$ denote the set of all stopping times with values in $(0,T]$.
The following is a localised version of Definition~\ref{def:Consistent}:
\begin{definition}\label{def:LocallyCTV}
A locally consistent total variance model is a couple $(S_{\cdot}, \omega_{\cdot}(k_{\cdot}, T))$ such that
there exists $\tau\in\Tt_T$ for which, up to~$\tau$:
\begin{enumerate}[(i)]
\item the process $S$ is a strictly positive $\QQ$-martingale;
\item for every $K>0$, the process $\omega_{\cdot}(k_{\cdot},T)$ has continuous paths and is 
strictly positive;
\item there is no Butterfly arbitrage.
\item for every $K>0$, the process~$C$ defined by 
$C_t :=S_t\BS\left(k_t, \sqrt{\omega_t(k_t,T)}\right)$ is a $\QQ$-martingale.
\end{enumerate}

We denote~$\Qql_T$ the set of all locally consistent total variance models,
and we say that there is no dynamic arbitrage on $[0, \ts]$ if~$\Qql_T$ is not empty with $\tau \in \Qql_T$.
\end{definition}

\begin{remark}
Here again we assume that the implied volatility of Calls and Puts is the same, 
so that the Put-Call-Parity holds by assumption, 
and we can equivalently replace Definition~\ref{def:LocallyCTV}(iv) 
by the martingale property of the Put price process, 
and Definition~\ref{def:LocallyCTV}(iii) by the corresponding properties of the Put price.
Also the martingale property of the option prices is equivalent here again to the local martingale property.
\end{remark}

Obviously, $\Qq_T \subset \Qql_T$, so that every consistent total variance model can be localised. 
Conversely, a local consistent total variance model might or might not be extended to a fully consistent one. Even if it might be,
this possibility could well be of no use in practice because of the complexity of the overall model. 
We believe therefore that this local approach has a strong interest on its own, 
and we shall also call such a locally consistent model a {\it bubble}.

If one trades only within a bubble lifespan~$[0, \tau]$,
with possibly additionally unwinding trades at the expiry~$T$, there is no arbitrage to be made. 
Arbitrage, if some, will follow from purely dynamic strategies involving unwinding positions beyond the bubble lifespan. 
We believe this set-up, albeit not surprising from a strict mathematical point of view, might be very relevant to account for real life joint underlying and options
dynamics. Indeed, if one goes back to the local dynamic of an SSVI bubble 
in Theorem~\ref{thm:thetaPhiDynamic}, we stress that the driving
Brownian motion of the SSVI parameters should be independent of the driving Brownian motion of~$S$, with no necessary relation to the Brownian motion (if any) driving
the instantaneous volatility process~$v$. 
This might correspond to the observations of the high frequency joint dynamics of the instantaneous volatility and of option prices
on Equity index options in~\cite{abergel2012drives}, where the complementary noise driving the non-Delta move of the option prices does not correspond to the idiosyncratic noise of the instantaneous volatility.
In fact, beyond this high-frequency situation, it might be the case that the joint dynamic of the underlying and of the Vanilla option prices is a succession of bubbles, with some bubbles
far from a fully consistent joint dynamic, and others closer.

\subsection{Symmetric SSVI smile bubbles}
In the symmetric SSVI dynamics, we proved that there cannot exist any consistent total variance model, by identifying local dynamics that $(\theta, \varphi)$ should satisfy for option
prices to be local martingales. 
The missing ingredient to design a bubble is the no-Butterfly condition; 
in the symmetric SSVI case, we have an explicit description of the no-Butterfly domain in terms of the parameters. 
So assuming the model parameters start within this domain, 
and making them evolve according to the local dynamics identified in our computations, 
we indeed obtain a bubble. 
The bubble lifespan is then the first time when either~$\theta$ or~$\varphi$ becomes negative, 
or when we exit the no-Butterfly domain. 
The explicit description of the no-Butterfly arbitrage region is given by Condition~\ref{cor:ButterflyArb}, 
so that combining it with~\eqref{eq:SolutionSystem}, 
$\psi_T^2 \leq \B(\theta_t)$ holds for all $t\in [0,T]$.
The function~$\B$ is smooth, strictly increasing from~$[0,4]$ to $[0, 16]$, and constant on~$[4, \infty)$,
so its inverse~$\Binv$ is well defined from~$[0,16]$ to~$[0,4]$,
and $\theta_0 \geq \Binv(\psi_T^2)$, implying that~$\theta$ is non-negative.
Therefore, if 
$\psi_T<4$,
then $\theta_0> \Binv(\psi_T^2)$, and we consider the stopping time 
\begin{equation}\label{eq:Tau}
\ts := \inf \left\{ t \leq T, \theta_t < \Binv(\psi_T^2)\right\} \in \Tt_{T}.
\end{equation}
On $[0, \ts]$ both Vanilla option prices and stock price are martingales, 
and the no-Butterfly condition holds, so that we have an explicit description of the symmetric SSVI implied volatility bubbles.
These however depend on the parameter~$\psi_T$ and on the dynamics of the process~$v$, 
which partially drives~$\theta$. Indeed, from ~\eqref{eq:SolutionSystem},
$$
\theta_t
 = \theta_0 + \left(\frac{\psi_T^2}{16} - 1\right) \int^t_0 v_u \D u + \psi_T \int^t_0 \sqrt{v_u}\,\D B^\theta_u,
$$
where~$B^\theta$ is a Brownian motion independent from~$B^S$, 
which may, or may not, be related to the Brownian motion driving~$v$.
We need to check that~$\theta$ remains positive, but assuming that~$\theta_0$ satisfies the no-Butterfly constraint in Condition~\ref{cor:ButterflyArb}, the positivity
of~$\theta$ before time~$\ts$ follows from its definition. 
Therefore any locally integrable instantaneous process independent from~$B^S$ defines a symmetric SSVI implied volatility bubble up to the stopping time~$\ts$. 
Note in passing that the Novikov condition on~$v$ grants the martingale property of~$S$, 
but could certainly be weakened in this bubble context.

Due to the finite lifespan of the bubble, shorter than the options' maturity, 
option prices are not given by the expectation of the terminal payoff 
under the risk-neutral dynamic of the underlying. 
It is natural in the context of a bubble to then take the conditional expectation at the bubble time boundary. 
Since option prices are true martingales, Doob's optional stopping theorem~\cite[Chapter II, Section 3]{Revuz},
implies that the price of the option at time~$t\leq \ts$ is
$\EE_t\left[S_\tau\BS(k_\tau, \varpi_\tau) \right]$,
where $S_. = K \exp{(-k_.)}$ and
$$
\varpi_t := \frac{1}{2}\left(\theta_{t}+\sqrt{\theta_{t}^2+\psi_T^2
k_{t}^2}\right).
$$
For $t\in [0,T]$, $\varpi_t$ depends on the terminal expiry~$T$ through~$\psi_t$.
Furthermore, the option price is also given by the Black-Scholes formula composed with SSVI, so that
(recall that $S_0=1$, hence $k_0=\log{K}$)
\begin{equation}\label{eq:BubbleEquation}
S_0\BS(k_0, \varpi_0) = \EE_0\left[S_\tau\BS(k_\tau, \varpi_\tau)\right].
\end{equation}
This equality holds irrespective of the dynamics of the variance process, even if the latter directly impacts 
the joint law of $(\tau, k_\tau, \theta_{\ts})$.  
Since the strike can be chosen arbitrarily, this in turn yields 
an interesting property of this joint law, of which we provide several examples below.
We christen~\eqref{eq:BubbleEquation} the {\it Bubble Master Equation}.
It will be convenient in some situations to go back to a Brownian setting by time-change: let us restart from
$$
\BS(k_0, \varpi_0) 
 = \EE_0\left[\ind_{\{\tau<T\}} \BS(k_\tau, \varpi_\tau)
 +  \ind_{\{\tau\geq T\}} \BS(k_T, \varpi_T)\right].
$$
By time change, there exist two Brownian motions $\beta^\theta, \beta^S$,
independent conditionally on~$(v_t)_{t\geq 0}$, such that
\begin{equation}\label{eq:DynQ}
\theta_t = \theta_0 + \left(\frac{\psi_T^2}{16}-1\right)t + \psi_T \beta^\theta_t,\qquad
S_t = S_0 \exp\left\{\beta^S_t - \frac{t}{2}\right\},\quad\text{and}\quad
\tau = \inf \left\{ r< \Vv_T, \theta_r<\Binv(\psi_T^2) \right\},
\end{equation}
where $\Vv_T:=\int_{0}^{T}v_s \D s$.
With the probability measure~$\QQ$ via $\frac{\D\QQ}{\D\PP_T} := \mathcal{E}(L_T)$,
where 
$L_T := \eta \beta^\theta_T$
and $\eta = -\frac{\psi_T}{16} + \frac{1}{\psi_T}>0$,
the process~$\widetilde{\beta}$ defined by 
$\widetilde{\beta}_t= -\beta^\theta_t + \eta t$ is a $\QQ$-Brownian motion by Girsanov's Theorem.
Furthermore, $\tau$ is now the first hitting time of~$\widetilde{\beta}$ of 
$a:=\frac{1}{\psi_T}\left(\theta_0-\Binv(\psi_T^2)\right)$, and hence, by conditioning,
\begin{equation}\label{eq:BSPrice}
\BS(k_0, \varpi_0)
= \EE_0^{\QQ}
 \left[\exp\left\{\eta \widetilde{\beta}_\tau -\frac{\eta^2 \tau}{2} \right\} \ind_{\{\tau<\Vv_T\}}
 \BS(k_\tau, \varpi_\tau)\\
 + \exp\left\{\eta \widetilde{\beta}_{\Vv_T} -\frac{\eta^2 \Vv_T}{2}\right\} \ind_{\{\tau\geq \Vv_T\}}
 \BS(k_{\Vv_T}, \varpi_{\Vv_T})\right].
\end{equation}

\subsubsection{The Black-Scholes-SSVI bubble}
We investigate the Bubble master equation when $T$ goes to infinity, choosing $\psi_T$ such that 
it is constantly equal to a given $\psi_\infty<4$. Assume $v_t = v_0>0$ almost surely for all $t\in [0,T]$;
then $\Vv_T = v T$ diverges to infinity as~$T$ increases.
In this limiting case, we deduce from~\eqref{eq:BSPrice}:
\begin{proposition}
For any $\psi_\infty<4, \theta_0>\Binv(\psi_\infty^2)$, $k_0\in\RR$, with $\bh:=\Binv(\psi_\infty^2)$, the equality
$$
 \BS(k_0, \varpi_0)
  = \frac{a \E^{\eta a}}{2 \pi} \int_{\RR}\E^{-\frac{y}{2}}  
\BS\left(k-y, \frac{1}{2}\left(\bh + \sqrt{\psi_\infty^2 (k -y)^2 + \bh^2}\right)\right)
\sqrt{\frac{4 \eta^2+1}{y^2+a^2}}
\K_{1}\left(\frac{\sqrt{\left(4 \eta^2+1\right) \left(y^2+a^2\right)}}{2}\right)\D y
$$
holds, where $\eta := \frac{1}{\psi_\infty} - \frac{\psi_\infty}{16}$,
$a:=\frac{\theta_0-\bh}{\psi_\infty} >0$, and 
$\K_{1}$ denotes the modified Bessel function of the second kind.
\end{proposition}

\begin{proof}
Since $a, \eta>0$, then $\eta \widetilde{\beta}_T<\eta a$, 
so that the second term in~\eqref{eq:BSPrice} tends to zero pointwise as~$T$ increases, 
while remaining bounded above by $K\E^{\eta a}$.
Since~$\tau$ is finite almost surely, the first term in~\eqref{eq:BSPrice} increases to 
$$
\exp\left\{\eta \widetilde{\beta}_\tau -\frac{\eta^2 \tau}{2}\right\} 
\BS(k_\tau, \varpi_\tau)\ind_{\{\tau<T\}},
$$ 
so, dominated convergence for the second term and monotone convergence for the first one imply
$$
\BS(k_0, \varpi_0)
  = \EE_0^{\QQ}\left[\exp\left\{\eta \widetilde{\beta}_\tau -\frac{\eta^2 \tau}{2}\right\} 
  \BS(k_\tau, \varpi_\tau)\right].
$$
Under~$\QQ$, using~\eqref{eq:DynQ}, we have $k_t = \log\left(\frac{K}{S_t}\right) = k-\beta^S_t+\frac{1}{2}$, so that
$$
\BS(k_0, \varpi_0)
 = \EE_0^{\QQ}\left[\E^{\eta a -\frac{1}{2} \eta^2 \tau} 
 \BS\left(k - \beta^S_\tau + \frac{\tau}{2}, \frac{1}{2}\left\{\bh + \sqrt{\psi_\infty^2 \left(k -\beta^S_\tau+\frac{\tau}{2}\right)^2 + \bh^2}\right\}\right)\right].
 $$
Under~$\QQ$, the density of~$\tau$ is
$p(t):=\frac{a}{\sqrt{2 \pi t^3}} \E^{-\frac{a^2}{2t}}$. 
Since~$\tau$ and~$\beta^S$ are independent, the right-hand side reads 
\begin{align*}
& \E^{\eta a} \int_{0}^{\infty} 
\left[\int_{\RR} \E^{ -\frac{1}{2} \eta^2 t} 
\BS\left(k - x\sqrt{t}+\frac{t}{2}, \frac{1}{2}\left\{\bh + \sqrt{\psi_\infty^2 \left(k -x\sqrt{t}+\frac{t}{2} \right)^2 + \bh^2}\right\}\right) p(t) \frac{\E^{-\frac{x^2}{2}}}{\sqrt{2 \pi}} \D x\right] \D t\\
 & = a \frac{\E^{\eta a}}{2 \pi} \int_{\RR} \E^{-\frac{y}{2}}  
\BS\left(k-y, \frac{1}{2}\left(\bh + \sqrt{\psi_\infty^2 (k -y)^2 + \bh^2}\right)\right) 
 \left(\int_{0}^\infty \E^{-\frac{1}{2}\left(\eta^2+\frac{1}{4}\right)  t} \E^{ -\frac{y^2+a^2}{2t}} \frac{\D t}{t^2}\right)\D y,
\end{align*}
with $y := x\sqrt{t}-\frac{1}{2} t$.
The proposition follows from the computation of the inner integral as
$$
\int_{0}^\infty \E^{-\frac{1}{2}\left(\eta^2+\frac{1}{4}\right)  t} \E^{ -\frac{y^2+a^2}{2t}} \frac{\D t}{t^2}
= \sqrt{\frac{4 \eta^2+1}{y^2+a^2}}
\K_{1}\left(\frac{\sqrt{\left(4 \eta^2+1\right) \left(y^2+a^2\right)}}{2}\right).
$$ 
 \end{proof} 
 
\begin{remark}
We can in fact apply exactly the same reasoning for any positive level~$\beta<\theta_0$ since
 we only use the martingale property of the price under the dynamic of~$\theta$: 
the same identity thus holds for any $0<\beta<\theta_0$, 
and not only for $\beta=\Binv(\psi_\infty^2)$. Observe that the right-hand side does not depend on~$\bh$, 
since the left-hand side does not. Interestingly enough, this equality does not seem easy to obtain by classical means.
 \end{remark}
 
We provide a numerical example to illustrate this Black-Scholes bubble,
starting from a symmetric smile given by~\eqref{eq:SSVIDyn}
with maturity $T=1$ and parameters $\theta(T) = \frac{1}{2}$ and $\varphi(T) = 1$, 
shown in Figure~\ref{fig:SSVISmile}.
We consider five Black-Scholes paths with $\sigma = 50\%$ 
over~$[0,T]$, with~$500$ equidistant time steps, plotted in Figure~\ref{fig:ThetaDynamics}, 
where we indicate the level $\Binv(\psi_T^2)= 0.013$. 
Once a path crosses this level, we show it dashed.
For each path, we further plot the corresponding total variance smiles in Figure~\ref{fig:ThetaDynamicsSmiles}.
The paths of~$\theta$ will either escape the no-Butterfly arbitrage domain 
(crossing~$\Binv(\psi_T^2)$) before~$T$, or reach the latter 
at a positive value, which is inconsistent in both cases. 
Yet, strictly up to this stopping time, the joint dynamics of the underlying and of the smile is  arbitrage-free, 
both from a static point of view and from a dynamic point of view, as there is no way to synthetise an arbitrage.

\begin{figure}[h!]
\includegraphics[scale=0.4]{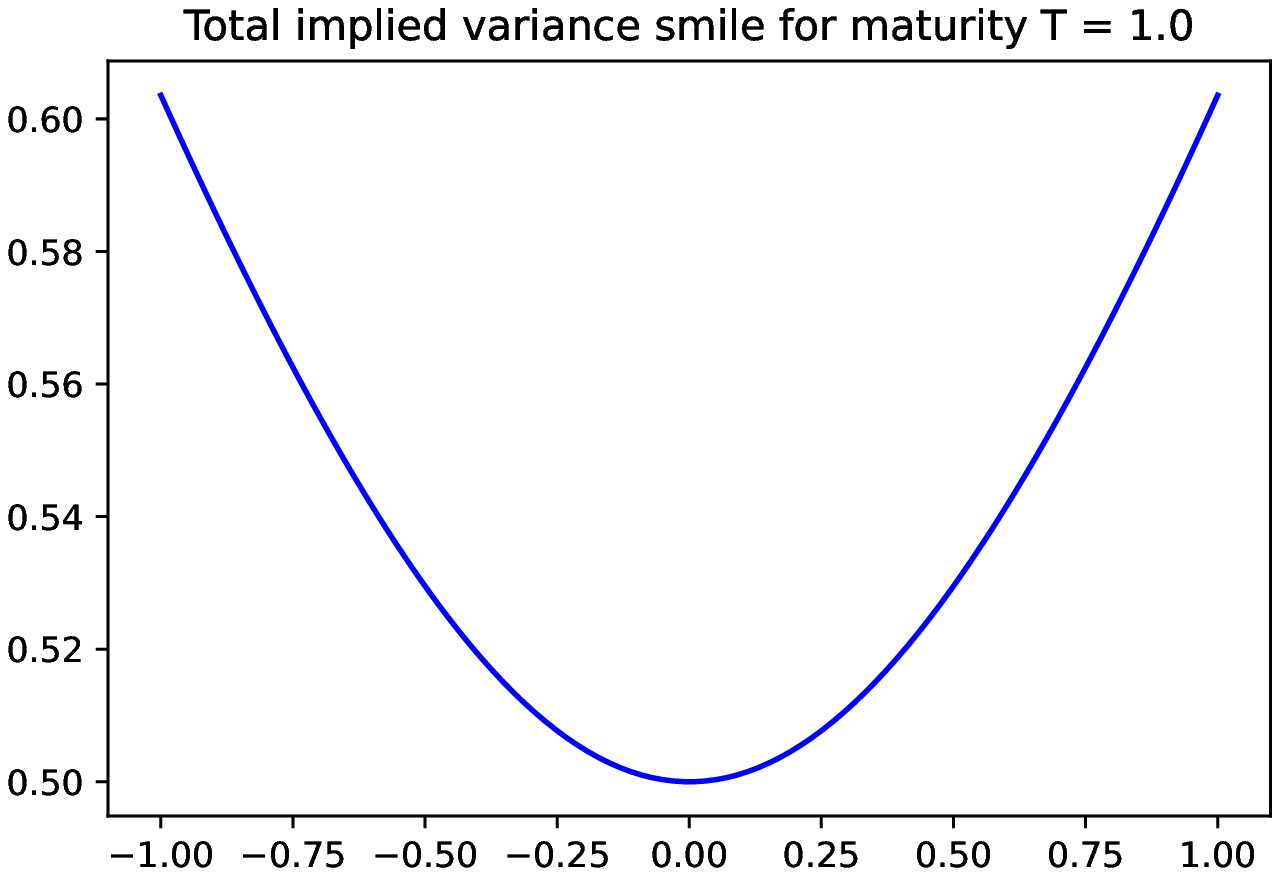}
\caption{Total implied variance smile with maturity $T=1$ year.}
\label{fig:SSVISmile}
\includegraphics[scale=0.4]{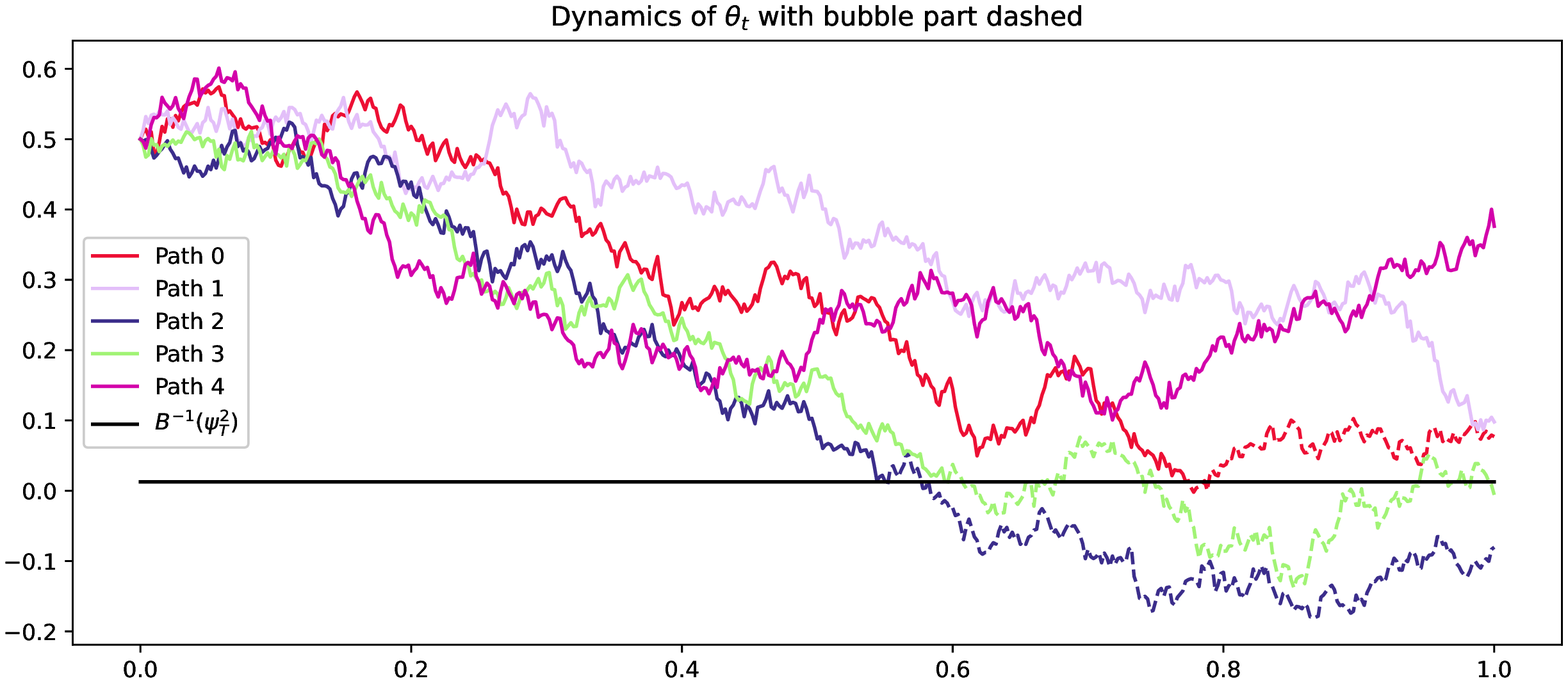}
\caption{Dynamics of the paths for $(\theta_t)_{t\in [0,1]}$.}
\label{fig:ThetaDynamics}
\end{figure}

\begin{figure}[h!]
\includegraphics[scale=0.4]{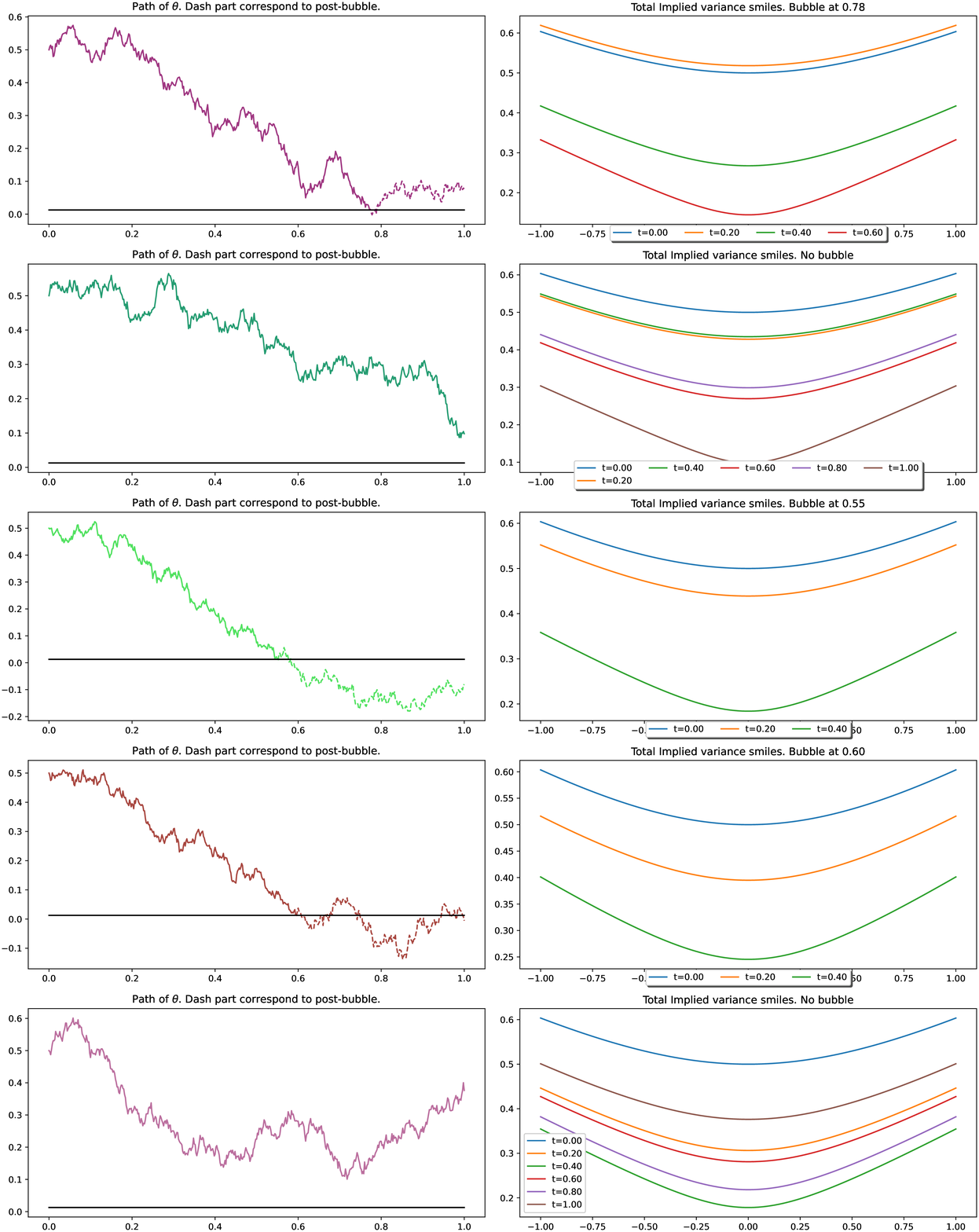}
\caption{Dynamics of the paths for $(\theta_t)_{t\in [0,1]}$, together with total implied variance smiles, at least up to the bubble time.}
\label{fig:ThetaDynamicsSmiles}
\end{figure}

\newpage
\subsubsection{The Heston-SSVI bubble}
We now investigate the existence of an SSVI bubble in the Heston model, 
that is when the variance process~$(v_t)_{t\geq 0}$ is a Feller diffusion of the form
$$
\D v_t = \kappa(\vb - v_t)\D t + \xi\sqrt{v_t}\,\D B^v_t,
\qquad v_0, \kappa, \vb, \xi>0.
$$
Yamada-Watanabe conditions~\cite[Proposition 2.13]{Karatzas}
guarantee a unique strong solution, and the Feller condition $2\kappa\vb\geq \xi^2$ 
that the latter never reaches the origin.
We consider two cases, depending on whether~$B^v$ and~$B^\theta$ are fully or anti correlated
(and independent of~$B^S$).
In the former case, \eqref{eq:SolutionSystem} reads
$$
\D\theta_t
 = \left(\frac{\psi_T^2}{16} - 1\right)v_t\D t - \psi_T \sqrt{v_t}B_t^\theta
 = \left(\frac{\psi_T^2}{16} - 1\right)v_t\D t - \frac{\psi_T}{\xi} \left(\D v_t - \kappa(\vb - v_t)\D t\right),
$$
so that
$\displaystyle
\theta_t = \theta_0 
+ \left(\frac{\psi_T^2}{16} - 1- \frac{\kappa\psi_T}{\xi}\right)\Vv_t - \frac{\psi_T}{\xi} (v_t - v_0) +\frac{\kappa\vb\psi_T t}{\xi}$.
When $\la B^\theta, B^v\ra_t = -\D t$, then 
$\displaystyle 
\theta_t = \theta_0 
+ \left(\frac{\psi_T^2}{16} - 1+ \frac{\kappa\psi_T}{\xi}\right)\Vv_t + \frac{\psi_T}{\xi} (v_t -v_0) - \frac{\kappa\vb\psi_T t}{\xi}$.

We can then proceed as in the Black-Scholes case, at least for the innermost conditional expectation; the subsequent computations are still more intricate, since in general the law of $\tau$ is unknown. The symmetric SSVI bubble equation
provides in this case some information of the joint law of $\tau, v_\tau$ and $\Vv_\tau$.

\section{Conclusion}
We showed that consistent joint dynamics of a stock price and symmetric SSVI smiles 
over a finite time horizon is infeasible, albeit with two key by-products:
first, a generalisation to non-symmetric smiles may potentially provide an answer, 
but the heavy computations require a dedicated analysis;
second, we identified a new regime of implied volatility bubbles, vanishing at maturity.
These bubbles are singular as arbitrage opportunities 
(with strategies including unwinding positions at maturity) cannot be achieved during the lifespan of the bubble. 
They might account for
some empirical findings in high-frequency data of Equity indices and options thereupon, 
as in~\cite{abergel2012drives}. 
One could even conjecture that those bubbles are in fact more a rule than an exception, 
with real-life dynamics over the lifetime of an option is a succession of such bubbles.

The next steps will be to investigate the present setup in a wider class of models, 
including correlation between the underlying stock and its driving volatility.
It indeed remains to be seen whether jump processes or rough volatility dynamics for example
may provide non-trivial results. 



\begin{thebibliography}{99}

\bibitem{abergel2012drives}F. Abergel and R. Zaatour.
What drives option prices?
\textit{Journal of Trading}, {\tt 7}(3): 12-28, 2012.

\bibitem{Babbar}K. Babbar.
Aspects of stochastic implied volatility in financial markets.
PhD Thesis, Imperial College London, 2001.

\bibitem{Black}F. Black and M. Scholes.
The pricing of options and corporate liabilities. 
\textit{Journal Political Econ.}, {\tt 81}(3): 637-654, 1973.

\bibitem{carr2014implied}P. Carr and J. Sun.
Implied remaining variance in derivatives pricing.
\textit{The Journal of Fixed Income}, Spring 2014.

\bibitem{CarrWu}P. Carr and L. Wu.
Analyzing volatility risk and risk premium in option contracts: a new theory.
\textit{JFE}, {\tt 120}: 1-20, 2016.

\bibitem{corbetta2019robust}P. Cohort, J. Corbetta, I. Laachir and C. Martini.
Robust calibration and arbitrage-free interpolation of SSVI slices.
\textit{Decisions in Economics and Finance}, 1-13, 2019

\bibitem{Davis}M.H.A. Davis and J. Ob{\l}{\'o}j.
Market completion using options.
Advances in Mathematics of Finance, {\tt 43}: 49-60, 2008.

\bibitem{DeMarco}S. De Marco and C. Martini.
Quasi-explicit calibration of Gatheral SVI model.
Zeliade Systems White Paper, available at
\href{http://www.zeliade.com/whitepapers/zwp-0005.pdf}{http://www.zeliade.com/whitepapers/zwp-0005.pdf}, 2009.

\bibitem{DupireVol}B. Dupire.
Pricing with a smile.
\textit{Risk}, 1994.

\bibitem{FouqueBook}J.P. Fouque, G. Papanicolaou and R. Sircar.
Derivatives in Financial markets with stochastic volatility.
CUP, 2000.

\bibitem{FGGJT}P.K. Friz, J. Gatheral, A. Gulisashvili, A. Jacquier and J. Teichmann.
Large deviations and asymptotic methods in Finance.
Springer International Publishing, 2015.

\bibitem{GatheralSVI}J.~Gatheral.
A parsimonious arbitrage-free implied volatility parameterization with application to the valuation of volatility derivatives.
Presentation at Global Derivatives, 2004.

\bibitem{GatheralBook}J.~Gatheral.
The volatility surface: a practitioner's guide.
John Wiley \& Sons, 2006.

\bibitem{SVIHeston}J. Gatheral and A. Jacquier.
Convergence of Heston to SVI.
\textit{Quantitative Finance}, {\tt 11}(8): 1129-1132, 2011.

\bibitem{SSVIGat}J. Gatheral and A. Jacquier.
Arbitrage-free SVI volatility surfaces.
\textit{Quantitative Finance}, {\tt 14}(1): 59-71, 2014.

\bibitem{SSVIGen}G. Guo, A. Jacquier, C. Martini and L. Neufcourt.
Generalised arbitrage-free SVI volatility surfaces. 
\textit{SIAM Journal on Financial Mathematics}, {\tt 7}(1), 619-641, 2016.

\bibitem{Hafner}R. Hafner.
Stochastic implied volatility: a factor-based model.
Springer, 2004 .

\bibitem{SABR}P.S. Hagan, D. Kumar, A. Lesniewski and D.E. Woodward.
Managing smile risk.
\textit{Wilmott}, 84-108, 2002.

\bibitem{eSSVI}S.~Hendriks and C.~Martini.
The extended SSVI volatility surface.
\textit{Journal of Computational Finance}, {\tt 22}(5): 25-39, 2019.

\bibitem{Heston}S.L. Heston.
A closed-form solution for options with stochastic volatility with applications to bond and currency options.
\textit{The Review of Financial Studies}, {\tt 6}(2): 327-343, 1993.

\bibitem{JMM17}A.~Jacquier, C.~Martini and A.~Muguruza.
On VIX futures in the rough Bergomi model.
\textit{Quant. Finance}, {\tt 18}(1): 45-61,~2018.

\bibitem{Jourdain}B.~Jourdain.
Loss of martingality in asset price models with lognormal stochastic volatility.
Preprint \textit{CERMICS}, 2004.

\bibitem{Karatzas}I. Karatzas and S. Shreve.
Brownian motion and stochastic calculus.
Springer-Verlag, New-York, 1988.

\bibitem{Kazamaki}N. Kazamaki.
Continuous exponential martingales and BMO.
Lecture Notes in Mathematics, {\tt 1579}, Springer-Verlag, 1994.

\bibitem{Ledoit}O. Ledoit and P. Santa-Clara.
Relative pricing of options with stochastic volatility.
Preprint \href{https://papers.ssrn.com/sol3/papers.cfm?abstract_id=121257}{SSRN:121257}, 1998.

\bibitem{Lyons}T. Lyons.
Derivatives as tradable assets.
Seminario de Matemática Financiera MEFF-UAM, {\tt 2}: 213-232, 1997.

\bibitem{arianna1st}C. Martini and A. Mingone
No arbitrage SVI.
Preprint \href{https://papers.ssrn.com/sol3/papers.cfm?abstract_id=3594528}{SSRN:3594528}, 2020.

\bibitem{sun2014tale}Q. Niu and J. Sun.
A tale in three cities: comparison between GVV, SVI and IRV.
Preprint \href{https://papers.ssrn.com/sol3/papers.cfm?abstract_id=2739302}{SSRN:2739302}, 2014.

\bibitem{Revuz}D. Revuz and M. Yor.
Continuous martingales and Brownian motion.
Springer-Verlag Berlin Heidelberg, {\tt 293}, 1999.

\bibitem{Romano-Touzi}M.~Romano and N.~ Touzi. 
Contingent claims and market completeness in a stochastic volatility model. 
\textit{Mathematical Finance}, {\tt 7}(4): 399-412, 1997.

\bibitem{SBucher}P.J. Sch\"onbucher. 
A market model for stochastic implied volatility. 
\textit{Phil. Trans. Royal Soc. A}, {\tt 357}(1758): 2071-2092, 1999.

\bibitem{schweizer2008arbitrage}M. Schweizer and J. Wissel.
Term structures of implied volatilities: Absence of arbitrage and existence results.
\textit{Mathematical Finance}, {\tt 18}(1): 77-114, 2008.

\bibitem{schweizer2008term}M. Schweizer and J. Wissel.
Arbitrage-free market models for option prices: the multi-strike case.
\textit{Finance and Stochastics}, {\tt 12}(4): 469-505, 2008.

\bibitem{TouziBook}N. Touzi.
Optimal stochastic control, stochastic target problems, and Backward SDE.
Springer-Verlag New York, 2013. 

\end{thebibliography}
\end{document}